\documentclass[format=acmsmall, review=false]{acmart}

\usepackage{acm-ec-18}
\usepackage{booktabs} 

\usepackage{dsfont}
\usepackage{xfrac}

\usepackage{accents} 

\newcommand{\lOutOfkTEXT}{\emph{$\ell$-out-of-$k$}}
\newcommand{\lOutOfkProphetTEXT}{\emph{$\ell$-out-of-$k$ prophet}}
\newcommand{\lOutOfkSecretaryTEXT}{\emph{$\ell$-out-of-$k$ secretary}}
\newcommand{\OneOutOfkProphetTEXT}{\emph{$1$-out-of-$k$ prophet}}

\newcommand{\Omit}[1]{}

\newcommand{\CR}[2]{{\color{red}{#1}\color{blue}{#2}}}

\newcommand{\Reals}{\mathbb{R}}

\newcommand{\SetSt}{\middle\vert}

\newcommand{\E}{\mathbb{E}}

\newcommand{\ALG}{\ensuremath{\mathtt{ALG}}}

\newcommand{\val}{v}
\newcommand{\sample}{s}
\newcommand{\vals}{{\bf v}}
\newcommand{\samples}{{\bf s}}
\newcommand{\argmax}{argmax}

\newcommand{\dist}{D}

\newcommand{\win}{W}

\newcommand{\Exp}{\mathop{\mathbb{E}}}

\newcommand{\Xth}[1]{#1^{th}}

\newcommand{\TOPSet}{TOP}
\newcommand{\TOPellSet}{{TOP}^\ell}
\newcommand{\TOPellVal}{{top}^\ell}
\newcommand{\TOPenVal} {{top}^n}
\newcommand{\R}{\Reals}
\newtheorem*{algorithm*}{Algorithm}
\begin{document}
\title[Prophets and Secretaries with Overbooking]{Prophets and Secretaries with Overbooking}


\author{Tomer Ezra}
\email{tomer.ezra@gmail.com}
\affiliation{Tel-Aviv University}
\author{Michal Feldman}
\email{mfeldman@tau.ac.il}
\affiliation{Tel-Aviv University}
\author{Ilan Nehama}
\email{ilan.nehama@mail.huji.ac.il}
\affiliation{Bar-Ilan University}

\begin{abstract}
The prophet and secretary problems demonstrate online scenarios involving the optimal stopping theory.
In a typical prophet or secretary problem, selection decisions are assumed to be immediate and irrevocable.
However, many online settings accommodate some degree of revocability.
To study such scenarios, we introduce the \lOutOfkTEXT{} setting, where the decision maker can select up to $k$ elements immediately and irrevocably, but her performance is measured by the top $\ell$ elements in the selected set.
Equivalently, the decision makes can hold up to $\ell$ elements at any given point in time, but can make up to $k-\ell$ returns as new elements arrive.

We give upper and lower bounds on the competitive ratio of $\ell$-out-of-$k$ prophet and secretary scenarios.
These include a single-sample prophet algorithm that gives a competitive ratio of $1-\ell\cdot e^{-\Theta\left(\frac{\left(k-\ell\right)^2}{k}\right)}$, which is asymptotically tight for $k-\ell=\Theta(\ell)$.
For secretary settings, we devise an algorithm that obtains a competitive ratio of $1-\ell e^{-\frac{k-8\ell}{2+2\ln \ell}} - e^{-k/6}$, and show that no secretary algorithm obtains a better ratio than $1-e^{-k}$ (up to negligible terms).
In passing, our results lead to an improvement of the results of~\cite{AssafSC2000}
for \OneOutOfkProphetTEXT{} scenarios.

Beyond the contribution to online algorithms and optimal stopping theory, our results have implications to mechanism design.
In particular, we use our prophet algorithms to derive {\em overbooking} mechanisms with good welfare and revenue guarantees; these are mechanisms that sell more items than the seller's capacity, then allocate to the agents with the highest values among the selected agents.

\end{abstract}

\keywords{Prophet inequality;
	Secretary problem;
	Online algorithms;
	Mechanism design;
	Welfare approximation}

\maketitle


\section{Introduction}
\label{sec:introduction}

Consider a scenario where a decision maker observes a sequence of $n$ non-negative real-valued awards, $\val_1, \ldots, \val_n$.
When an algorithm, denoted by \ALG, serving on behalf of the decision maker reaches the $i^{th}$ award, $\val_i$, 
it needs to make an immediate and irrevocable decision whether or not to accept the award.
If it accepts $v_i$, the game terminates with an award of $\val_i$; otherwise, it continues to the next round, and the award $\val_i$ is lost forever.
The performance of an algorithm for such an online scenario is often measured by the {\em competitive ratio}, defined as the worst case ratio between the award accepted  by \ALG{} and the maximal award in hindsight.

Without imposing additional assumptions on the input, the competitive analysis framework produces no insights about the design of algorithms for such scenarios.
In two natural frameworks, however, known as the {\em prophet} and the {\em secretary}, good guarantees can be given.

\paragraph{Prophet.}

Prophet inequalities refer to guarantees on the competitive ratio in a setting, where every award $v_i$ is drawn independently from a known distribution $\dist_i$. The competitive ratio is the ratio between the expected performance of the algorithm and the expected maximal award, where the expectation is taken over the product distribution $\dist = \dist_1 \times \ldots \times \dist_n$, and possibly also over the random coin tosses of \ALG{}, if it is randomized.

Krengel, Sucheston, and Garling~\shortcite{KrengelS77,KrengelS78} demonstrated the first fundamental result in this framework, proving the existence of an algorithm with a competitive ratio at least a half. In other words, a ``prophet" who observes the entire sequence of realized awards from the outset, and simply takes the maximal award when reached, can gain at most twice the value that a gambler, who makes immediate and irrevocable decisions based on present and past observations only, can gain.
Samuel-Cahn~\shortcite{Samuel84} later showed that this competitive ratio can be obtained with a simple threshold algorithm: setting some threshold and accepting the first award that exceeds it.

\paragraph{Secretary.}

In the secretary framework, the awards are chosen by an adversary, but the arrival order is assumed to be randomly and uniformly distributed. The performance of \ALG{} is taken in expectation over the random arrival order, and possibly also over the random coin tosses of \ALG{}, if it is randomized. This performance is measured with respect to the maximal award in the sequence.
This framework appeared first in Martin Gardner's \emph{Scientific American} column~\cite{Gardner60}.
Gilbert and Mosteller~\shortcite{GilbertMosteller} presented an algorithm that achieved an asymptotic competitive ratio of $\sfrac{1}{e}$ for this problem, that is, an online algorithm that picks the maximal award with a probability at least $\sfrac{1}{e}$. This bound is asymptotically tight (Note that since the values are chosen by an adversary, bounding the probability of picking the maximal element is equivalent to bounding the ratio between the picked element and the maximal element).

\vspace{0.1in}

The two frameworks have significant implications for the design and analysis of auction and posted price mechanisms.
In recent years there has been a surge of interest in applying results from secretary and prophet scenarios to mechanism design settings.
For example, a direct implication of the classical prophet inequality is that a seller who knows the distributions from which buyers' values are drawn can achieve half of the optimal welfare by posting a single price and selling to the first buyer whose value exceeds the price.
Similarly, a direct implication of the classical secretary setting is an auction setting that obtains $1/e$ of the optimal welfare: sample at random a $(1/e)$ fraction of the buyers and query their value, then set a price at the maximal sampled value and offer this price to the buyers in a random order.
Although welfare implications are more direct, these results also lead to strong approximation guarantees on revenue, mainly in single parameter settings, but also in multi-parameter settings, such as matching markets \cite{ChawlaHMS10,AlaeiHNPY15,ChawlaHK07,KleinbergW12,AzarKW14}.

Following these observations, a series of works have looked at generalizations of prophet and secretary.
The common scenario studied in these works is the following: a collection $F$ of feasible sets of elements is given, the values of the elements are revealed one by one in an online fashion, the decision maker makes an immediate and irrevocable decision whether to accept or reject each element (under the constraint that the set of selected elements must belong to $F$ at all times), and the value of an accepted set is the sum of values of elements in the set.
%
%
An example of a feasibility constraint is a cardinality constraint, where one can accept at most $k$ elements.
This corresponds to selling $k$ identical items.
It was shown that prophet and secretary with cardinality $k$ admit competitive ratios of $1-1/\sqrt{k+3}$ and $1-5/\sqrt{k}$, respectively \cite{Kleinberg2005, HajiaghayiKS07,Alaei14}.

Additional feasibility constraints include independent sets in matroids~\cite{KleinbergW12,BabaioffIK07,BabaioffIKK08}, polymatroids~\cite{DuettingK15b}, knapsack constraints \cite{FeldmanSZ15,DuettingFKL17}, and general downward-closed~\cite{Rubinstein16}.
Recently, Duetting et al.~\shortcite{DuettingFKL17} established a new framework that uses prophet inequalities reasoning to establish pricing mechanisms with good welfare guarantees in combinatorial settings.


\vspace{0.1in}
\noindent {\bf Online scenarios with limited returns.}
All previous works considered scenarios in which one makes immediate and irrevocable selection decisions.
In many scenarios of interest, however, decisions are not fully irrevocable.
In this work we provide a modeling framework for the study of online scenarios with returns, parameterized by the number of possible returns.
We quantify the improvement in the performance of secretary and prophet algorithms as a function of this parameter.

In particular, we consider settings, termed \lOutOfkTEXT{}, in which the decision maker derives value from $\ell$ elements (e.g., selling $\ell$ identical items), but during the online process, he can select $k > \ell$ elements.
In other words, as elements arrive in an online fashion, the decision maker selects up to $k$ elements, but performance is measured by the top $\ell$ elements in the selected set.
This setting is mathematically equivalent to a setting in which the decision maker holds at most $\ell$ elements at all times, but can make up to $k-\ell$ returns as the values of the elements are revealed.

Different variants of prophet and secretary problems from the literature are special cases of this model.
For example, the classical secretary and prophet problems correspond to the special case where $k=\ell=1$.
Similarly, the secretary and prophet problems with cardinality constraints correspond to the special case where $\ell = k \geq 1$.

As the parameter $k$ varies, the model moves gradually from an offline to an online setting.
The pure offline setting is the special case $k=n$, where the decision maker chooses the highest $\ell$ values after observing the entire sequence.
The pure online setting is the special case $k=\ell$, where no returns can be made.
Thus, our study measures the change in the performance of the algorithm as the setting moves gradually from online to offline.

Scenarios of limited returns arise in various settings, including hiring, mechanism design, and job scheduling.
For example, sellers may have the option to commit to selling more items than they actually have, taking into account a fine they may need to pay if they fail to deliver.
Although the opportunity of returns comes with clear benefits, it also bears some costs, which are known to the designer of each setting.
Our results help a decision maker to quantify the benefits due to returns, providing her better tools to weigh the cost of returns against their benefits.
We proceed with several examples of economic scenarios exhibiting limited returns.

\paragraph{Hiring.} An employer who interviews candidates for a job that requires $\ell$ employees may be able to tentatively accept a slightly larger number of employees, then choose the best among them.
Extra hires, however, may be costly because of regulation and reputation effects.

\paragraph{Overbooking.}
An air carrier that has a capacity of $\ell$ seats for a given flight may wish to {\em overbook}, that is, sell more tickets than the capacity of the plane.
In the end, the company can use different methods to sell the tickets to the $\ell$ passengers with the highest value, e.g., offer a monetary compensation to agents who agree to postpone their flight.

\paragraph{Preemption in job scheduling}
A scheduler who needs to choose $\ell$ jobs to process may preempt existing jobs as new ones (perhaps more urgent) arrive. Preemption, however, may incur some cost related to maintaining the state of the system or the database.

In this work we study prophet and secretary settings with limited returns under cardinality constraints.
Prophet and secretary settings with some level of revocability is relevant under other feasibility constraints as well.
For example, in matroid prophet and secretary problems~\cite{KleinbergW12,BabaioffIK07,BabaioffIKK08}, a decision maker selects elements online, so that the selected set is an independent set of the matroid at all times, and the performance is measured by the sum of the selected values.
One can consider some leeway, given in the form of a second matroid.
Here, the elements selected by the decision maker should be an independent set $S$ of the second matroid at all times, but the performance is measured by the set $T$, such that $T$ is the independent set in the first matroid of maximal value that is contained in $S$.
This problem and similar ones remain as interesting future work. %
\subsection{Our Model}
\label{sec:model}

We consider a setting where \ALG{} can choose up to $k$ elements immediately and irrevocably, but the valuation of \ALG{}'s output  is then the sum of the top $\ell$ elements out of the $k$ chosen elements.
Let $\win$, $|\win| \leq k$, denote the set of elements chosen by \ALG, and let $\win' \subseteq \win$ be the set of the top $\ell$ values in $\win$. Then, the value of \ALG{}  is $\sum_{i \in \win'} \val_i$.

Note that this scenario is entirely equivalent to a setting in which the decision maker can hold up to $\ell$ elements at any point in time, but can replace an already chosen element with a new one up to $k - \ell$ times.
Indeed, an algorithm for the latter scenario can be simulated by the former scenario by selecting all elements that were selected along the process, including the ones replaced.
Similarly, an algorithm for the former scenario can be simulated by the latter one by selecting the same $k$ elements, and dropping the ones with the lowest value when capacity exceeds $\ell$.

In the remainder of this paper\CR{,}{} we use the former formulation of selecting $k$ elements and deriving value from the top $\ell$ in the selected set.
We refer to this setting as the \lOutOfkTEXT{} setting.

In the \lOutOfkProphetTEXT{} model, the performance of \ALG{} is taken in expectation over the product distribution $\dist = \dist_1 \times \ldots \times \dist_n$ (and possibly the randomness of the algorithm).
The performance of the prophet is given by $OPT(D) = \Exp_{\vals \sim D}[\max_{L \subseteq [n], |L|\leq \ell}\sum_{i \in L} \val_i]$.

In the \lOutOfkSecretaryTEXT{} model, the performance of \ALG{} is taken in expectation over the random order of arrival (and possibly the randomness of the algorithm), when considering worst case values $\vals = (\val_1, \ldots, \val_n)$.
This performance is measured with respect to the following benchmark:
$OPT(\vals) = \max_{L \subseteq [n], |L|\leq \ell}\sum_{i \in L} \val_i$.



Given a set $S$ of non-negative numbers, let $\TOPellSet(S)$ denote the subset of size $\ell$ of maximum total value, i.e., $\TOPellSet(S)= \argmax_{S' \subseteq S, |S'| \leq \ell}  \sum_{v \in S'} v$.
We extend this definition to a distribution $F$ over sets, and define the random variable $\TOPellSet(F) = \TOPellSet(S)$, where $S$ is drawn according to $F$.
%
The total value in the set $\TOPellSet(S)$ is denoted by $\TOPellVal(S)$, i.e., $\TOPellVal(S)=\sum_{v \in \TOPellSet(S)}v$.\footnote{%
	Note that although there might be several maximizers, so $\TOPellSet(S)$ might not be unique, $\TOPellVal(S)$ is well-defined. For simplicity, throughout this paper we assume that $\TOPellSet\left(S\right)$ is a maximizer (and not the set of maximizers) and notice that our statements do not depend on the
	chosen maximizer.}
Similarly, given a distribution $F$ over sets, let $\TOPellVal(F)=\Exp_{S \sim F}\left[\TOPellVal\left(S\right)\right]$.

An algorithm $\ALG$ for the \lOutOfkProphetTEXT{} setting 
selects up to $k$ elements in an online fashion.
Let $\ALG(\dist)$ denote the distribution over sets (of cardinality up to $k$) returned by $\ALG$, where the input is distributed according to $\dist$.
We say that an algorithm $\ALG$ induces an \lOutOfkProphetTEXT{} inequality with a competitive ratio $\rho$ if for every product distribution $\dist$ it holds that
$$ \TOPellVal(\ALG(\dist)) \geq \rho \cdot \TOPellVal(\dist).$$

As in~\cite{AzarKW14}, we also consider settings in which $\ALG$ has only limited information about the product distribution $\dist$.
A {\em single-sample} algorithm is one that has no information about the distribution $\dist$, except for a single sample from it. A single-sample algorithm $\ALG$ receives as input two vectors: (a) a vector of $n$ {\em samples}, distributed according to $\dist$, which is received offline, and (b) a vector of $n$ {\em values}, distributed according to $\dist$, which is received in an online fashion.
A single-sample algorithm $\ALG$ for the \lOutOfkProphetTEXT{} setting selects up to $k$ elements (drawn from $\dist$) in an online fashion, knowing the sample vector from the outset.
We denote the distribution over sets (of cardinality up to $k$) returned by $\ALG$ by $\ALG(\dist^2)$.
A single-sample algorithm $\ALG$ is said to induce an \lOutOfkProphetTEXT{} inequality with a competitive ratio $\rho$ if
$$ \TOPellVal(\ALG(\dist^2)) \geq \rho \cdot \TOPellVal(\dist).$$

A threshold algorithm \ALG{} for the \lOutOfkProphetTEXT{} setting, after observing some of the input elements but before accepting any of them, decides on a threshold $T$; After  this decision was made, \ALG{} accepts the first $k$ elements for which $\val_i > T$.
We see this single decision property as a desired simplicity property, and note that all the \lOutOfkProphetTEXT{} algorithms we present here are threshold algorithms.
A specific case of the above two is \emph{single-sample threshold} algorithms, in which the threshold $T$ is a function of the (offline) sample vector.

Note that all our algorithms are \emph{order oblivious}. Therefore, the same guarantees hold when an adversary chooses separately the order of samples and the order of values.

An algorithm $\ALG$ for the \lOutOfkSecretaryTEXT{} setting selects up to $k$ elements in an online fashion, where the values $\vals=(\val_1, \ldots, \val_n)$ are chosen by an adversary, but the values arrive at a uniformly random order.
Let $\ALG(\vals)$ denote the distribution over sets (of cardinality up to $k$) returned by $\ALG$.
We say that an algorithm $\ALG$ induces an \lOutOfkSecretaryTEXT{} inequality with a competitive ratio $\rho$ if for every vector of values $\vals$ it holds that
$$ \TOPellVal(\ALG(\vals)) \geq \rho \cdot \TOPellVal(\vals).$$
\subsection{Our Results and Techniques}
In this section we state our lower and upper bounds on the competitive ratios for the \lOutOfkTEXT{} prophet and secretary problems. These results are summarized in Tables \ref{t:prophet} (prophet settings) and \ref{t:secretary} (secretary settings), where the numbers in parentheses refer to the corresponding theorem numbers.

\begin{table}[h]
	\begin{center}
		\centering
		\begin{tabular}{|c|c|c|c|c|}
			\hline
			& Lower bound  & Upper bound \\ \hline
			Single-sample algorithm & $1-\ell\cdot 			e^{-\Theta\left(\frac{\left(k-\ell\right)^2}{k}\right)}$ (\ref{thm:alg-performance}) & 	$1 -\frac{2^{-(2k+1)}}{k+1}$ (\ref{thm:single_sample_lower}) \\ \hline
			$\dist_{max}$ algorithm & $1-\frac{3}{2}\cdot e^{-\sfrac{k}{6}}$ (\ref{thm:alg_max})
			& $1 - \frac{1}{(2k+2)!}$ (\ref{thm:prophet_example}) \\ \hline
		\end{tabular}			
		\caption{Results for \lOutOfkProphetTEXT settings.}
		\label{t:prophet}
		\centering
		\begin{tabular}{|c|c|c|c|c|}
			\hline
			& Lower bound  & Upper bound \\ \hline
			Secretary &
			$1-\ell e^{-\frac{k-8\ell}{2+2\ln \ell}} - e^{-k/6}$ (\ref{thm:secretary_b_new}) & 	
			$1-\frac{1}{e^k}+\frac{2}{3n}$ (\ref{thm:sec_lower}) \\ \hline
		\end{tabular}
		\caption{Results for \lOutOfkSecretaryTEXT settings.}
		\label{t:secretary}
    \label{t:sum}
	\end{center}
\end{table}

\vspace{0.1in}
\noindent{\bf Prophet Scenarios.}
\vspace{0.1in}

Our first theorem (Theorem \ref{thm:alg-performance}) provides a single-sample algorithm for the prophet setting.

\vspace{0.1in}
\noindent {\bf Theorem [single-sample, prophet, positive result]:}
There exists a single-sample algorithm that induces an \lOutOfkProphetTEXT{} with competitive ratio
$1-\ell\cdot e^{-\Theta\left(\frac{\left(k-\ell\right)^2}{k}\right)}$.
\vspace{0.1in}

The algorithm is a single-threshold algorithm, which sets a threshold that equals the $(\frac{\ell+k}{2})^{th}$ highest sample, and accepts all values exceeding this threshold, up to reaching capacity $k$.
We prove that this threshold algorithm exhibits two properties. First, it chooses all the highest $\ell$ values with high probability. Second, we show that for every product distribution, if the algorithm does not choose the $\ell$ highest values, it does not lead to a large loss.
On the other hand, we show (in Theorem \ref{thm:single_sample_lower}) that this result is tight if the number of possible returns is linear in $\ell$ (i.e., $k-\ell =\Theta(\ell)$).\footnote{When the number of possible returns is small ($k-\ell<C\cdot \sqrt{\ell\log\ell}$, for some constant $C>0$), our result is not tight; in fact, the original algorithm of \cite[Thm.~4.8]{Alaei14} (without returns) gives a better bound.}


\vspace{0.1in}
\noindent {\bf Theorem [single-sample, prophet, negative result]:}
No single-sample algorithm obtains a competitive ratio higher than $1 -\frac{2^{-(2k+1)}}{k+1}$.
\vspace{0.1in}

To establish the negative result, we first observe that a single-sample algorithm must select every value that significantly exceeds all samples and all previously observed values, so as to obtain a good competitive ratio.
Therefore, no single-sample algorithm can handle a scenario in which there are more than $k$ values with this property.
To conclude the result, we establish a product distribution that exhibits this scenario with sufficiently high probability.

We next present a deterministic algorithm for the \OneOutOfkProphetTEXT{} setting that improves a result obtained by Assaf and Samuel-Cahn~\shortcite{AssafSC2000}.
Assaf and Samuel-Cahn study a setting where instead of running a single online algorithm on the sequence, one is allowed to run $k$ simultaneous online algorithms, each selecting a single element; the performance of the algorithm is then measured by the maximum value obtained by any of the algorithms running in parallel.
They establish the existence of an algorithm that obtains a competitive ratio of $1- \sfrac{1}{(k+1)}$.
The information held by the algorithm is the distribution of the maximum value, $\dist_{max}$.

We improve their result by presenting an algorithm which obtains a competitive ratio of
$1-\frac{3}{2}\cdot e^{-\sfrac{k}{6}}$,
improving the competitive ratio from decreasing proportionally to $1/k$ to decreasing exponentially in $k$.
Our algorithm has the same information assumed in \cite{AssafSC2000}, namely the distribution $\dist_{max}$.
We obtain our result by providing a deterministic single-threshold algorithm for the \OneOutOfkProphetTEXT{} setting.
We then observe that any such algorithm can be simulated by $k$ simultaneous online algorithms, each selecting a single value.

The \OneOutOfkProphetTEXT{} problem was also studied in~\cite{Assaf2002,ASSAF2005127}.
In these works, the authors establish lower bounds on the competitive ratio that can be achieved 
	for any $k$~\cite{Assaf2002} and for the case $k=n-1$~\cite{ASSAF2005127}.
Their bounds are given as a recursive formula (and not in closed form), and are mostly approximated using numeric methods.
In contrast, we give a closed form lower bound that decreases exponentially with $k$.


\vspace{0.1in}
\noindent {\bf Theorem [$\dist_{max}$, prophet, positive result]:}
There exists a deterministic single-threshold algorithm that induces a \OneOutOfkProphetTEXT{} with a competitive ratio
	 $1-\frac{3}{2}\cdot e^{-\sfrac{k}{6}}$.
This algorithm knows only $\dist_{max}$.
\vspace{0.1in}

The algorithm sets a single threshold $T$, such that $\Pr_{x \sim \dist_{max}}[x < T] = \left(\frac{2}{3}\right)^{k-1}$, and selects the first $k$ elements that exceed $T$. To establish the result, we show that the probability that more than $k$ values exceed $T$ is negligible, and that in the event that more than $k$ values exceed $T$, we do not lose much.

Finally, Theorem \ref{thm:prophet_example} establishes a negative result for every prophet algorithm, even one that has full information on $\dist$ from the outset.

\vspace{0.1in}
\noindent {\bf Theorem [prophet, negative result]:}
No \lOutOfkProphetTEXT{} algorithm achieves a better competitive ratio than $1 - \frac{1}{(2k+2)!}$.
\vspace{0.1in}



The negative result is obtained by setting each distribution $\dist_i$ as having value $\frac{1}{p_i}$ with probability $p_i$ and value $0$ otherwise. We carefully set the values of $p_i$ to ensure that the best online algorithm always accepts each non-zero value. The asserted bound is then obtained by quantifying the loss of such an algorithm if there are more than $k$ non-zero values.



\vspace{0.1in}
\noindent{\bf Secretary Scenarios.}
\vspace{0.1in}

Theorem \ref{thm:secretary_b_new} provides an algorithm for \lOutOfkSecretaryTEXT{} settings.

\vspace{0.1in}
\noindent {\bf Theorem [secretary, positive result]:}
There exists an algorithm that induces an \lOutOfkSecretaryTEXT{} with a competitive ratio $1-\ell e^{-\frac{k-8\ell}{2+2\ln \ell}} - e^{-k/6}$.
\vspace{0.1in}

The algorithm divides the values into $\ell+1$ segments, numbered from $0$ to $\ell$.
In the $\Xth{j}$ segment the algorithm accepts $\val_i$ if it belongs to the $j$ highest values seen so far, and the capacity $k$ is not exhausted. We bound the probability that an element which belongs to the top $\ell$ elements is accepted by our algorithm. Our algorithm has two potential sources of loss, namely (a) elements that belong to OPT, but sufficiently many higher elements appeared before them, and (b) elements that belong to OPT that appear after the algorithm has exhausted its capacity. We bound the two losses to obtain our result.


On the negative side, we establish the following result:

\vspace{0.1in}
\noindent {\bf Theorem [secretary, negative result]:}
No \lOutOfkSecretaryTEXT{} algorithm achieves a better competitive ratio than $1-\frac{1}{e^k} + \frac{2}{3n}$.
\vspace{0.1in}

To establish this result we construct a sequence where the highest value is significantly larger than the second highest one, and show that no algorithm can choose the highest value with a probability higher than $1-\frac{1}{e^k} + \frac{2}{3n}$.

\vspace{0.1in}
\noindent{\bf Implications to mechanism design}
\vspace{0.1in}

In Section \ref{sec:mechanism-design}, we derive from our \lOutOfkProphetTEXT{} algorithms implications to mechanism design for economic scenarios with overbooking.
We show that given a single-threshold  \lOutOfkProphetTEXT{} algorithm with competitive ratio $\rho$, we can construct a truthful mechanism that obtains the same competitive ratio with respect to optimal welfare.
The mechanism is composed of two phases.
In the first phase, $k$ tickets for participating in the second phase are offered to buyers whose value exceeds a uniform threshold~$T$. The second phase is a VCG mechanism with reserve price~$T$.
Similarly, if all values are identically and independently distributed according to a regular distribution, we construct a two-phase mechanism that obtains the same competitive ratio with respect to revenue.

%
\subsection{Related Work}\label{sec:related_work}
Following the seminal works on prophet~\cite{KrengelS77,KrengelS78,Samuel84} and secretary~\cite{Gardner60,GilbertMosteller} problems, many new variants of these problems were studied.
We hereby mention a few variants beyond the ones already discussed in Section~\ref{sec:introduction}.
Krieger and Samuel-Cahn~\shortcite{Krieger2012} analyzed secretary scenarios with noisy values.
Azar et al.~\shortcite{AzarKW14} studied scenarios in which the decision maker does not know the distributions from which values are drawn, rather she receives a single sample from each distribution.
Vardi~\shortcite{Vardi15} analyzed secretary scenarios in which every element repeats several times.
Secretary settings with non-uniform random arrival orders were investigated by Kesselheim et al.~\shortcite{KKN15}.
Kesselheim et al.~\shortcite{DBLP:journals/corr/KesselheimT16} considered secretary settings in which commitments are temporary and the number of parallel commitments is bounded.
More recently, prophet secretary scenarios were introduced and studied~\cite{
	EsfandiariHLM17,
	EhsaniHKS18,
	DBLP:journals/corr/AzarChiplunkarKaplanArXiv
}.
These are scenarios in which the decision maker knows the distributions from which values are drawn (as in prophet scenarios) and the elements arrive in a random order (as in secretary scenarios).

Common to all of these works (as well as to the ones surveyed in Section~\ref{sec:introduction}) is that the objective function which the decision maker wishes to maximize is the sum of the selected elements. A different branch of generalizations considers other objective functions.
Gusein-Zade~\shortcite{Gusein-Zade1966}, Gilbert and Mosteller~\shortcite{GilbertMosteller}, and Frank and Samuels~\shortcite{Frank1980} analyze secretary problems in which the decision maker wishes to maximize the probability of choosing one of the maximal $k$ elements;
and Chow et al.~\shortcite{Chow1964} and Krieger and Samuel-Cahn~\shortcite{Krieger2009} consider the expected rank of the selected element.
Recently, motivated by questions in mechanism design, submodular and additional combinatorial valuations were also studied, for both secretary scenarios~\cite{FeldmanNS11,BarmanUCM12,BateniHZ13,FeldmanZ15,FeldmanI17} and prophet scenarios~\cite{RubinsteinS17}.

The literature on secretary and prophet inequalities has interesting implications to the design of online mechanisms and posted price mechanisms in particular.
These include settings with unit-demand valuations~\cite{ChawlaHMS10,Alaei14}, subadditive valuations~\cite{RubinsteinS17}, and more general combinatorial valuations~\cite{Alaei14,FeldmanGL15,DuettingFKL17},
For further details, see the recent survey by Lucier~\shortcite{Lucier17}.
Finally, Abolhassani et al.~\shortcite{AbolhassaniEEHKL17} study the classic prophet inequality setting in large markets, assuming random or best arrival order.

The study of online settings with some degree of revocability has been also considered by \citet{DBLP:conf/sigecom/BabaioffHK09},
\citet{DBLP:conf/wine/AshwinkumarK09}, and \citet{DBLP:conf/icalp/Varadaraja11}.
These papers consider the \emph{buyback problem}, where admitted elements can be revoked, but cancellation incurs some cost, and the goal is to maximize the net benefit.
In our work, we do not model the cost explicitly; rather, we consider scenarios in which a limited number of cancellations is permitted.


\newcommand{\DoNotCapitalizeK}{k}
\section{$\ell$-out-of-$\protect\DoNotCapitalizeK$ Prophet Inequalities}
\subsection{Single-Sample Algorithm}
\label{sec:single_sample}%
%
%
Consider the following family of single-sample threshold algorithms, parameterized by $\tau$. 
The algorithm $\ALG^{\tau}$ receives offline $n$ samples $\samples=\left<\sample_1, \ldots, \sample_n\right> \sim \dist$, and $n$ awards $\vals=\left<\val_1, \ldots, \val_n\right> \sim \dist$ online.
$\ALG^{\tau}$ picks at most $k$ awards online.



\begin{algorithm*}[$\ALG^\tau\left(\sample_1,\ldots,\sample_n,\val_1,\ldots,\val_n\right)$]\mbox{ }
	\begin{itemize}	
		\item \textbf{Offline}:\begin{tabular}[t]{l@{}}
		Choose a random permutation $\sigma$ uniformly from $\mathbb{S}_{2n}$.\\
		Let $(T,\sigma_T)$ be the $\tau$-maximal element of $\{(s_i,\sigma(i))\}_{i=1}^n$ (sorted by lexicographic order;\\ i.e., $(s_i,\sigma(i)) < (s_j,\sigma(j))$ if $s_i < s_j$ or $s_i=s_j$ and $\sigma(i) < \sigma(j)$.)\end{tabular}
		\item \textbf{Online}:
		Accept the first $k$ elements $i$ such that $(\val_i,\sigma(i)) > (T,\sigma_T)$.
	\end{itemize}
\end{algorithm*}

Remark: Note that the permutation $\sigma$ is used for tie breaking.
If $D_i$ is atomless for all $i$, then there is no need for tie breaking and the algorithm simply chooses the $\tau$-maximal sample $T$ in the first phase, and accepts the first $k$ awards that exceed $T$ in the second phase.

\begin{theorem}	\label{thm:alg-performance}
		For every product distribution $\dist$ and $\samples,\vals \sim \dist$:
	The expected sum of the top $\ell$ elements returned by $\ALG^{\left(\ell+k\right)/2}\left(\samples,\vals\right)$ is at least
		$1-\ell\cdot e^{-\Theta\left(\frac{\left(k-\ell\right)^2}{k}\right)}$
	of the expected sum of the maximal $\ell$ elements of  $\vals$. I.e., $\ALG^{\left(\ell+k\right)/2}$ achieves a competitive ratio of $1-\ell\cdot e^{-\Theta\left(\frac{\left(k-\ell\right)^2}{k}\right)}$.
\end{theorem}

\newcommand{\DX}{\dist_X}\newcommand{\DXaward}{\dist_X^{\text{award}}}%
For proving this theorem, we first define an auxiliary (non-product) distribution, $\DX$ for $X\in\Reals^{2n}$, and show it is sufficient to prove the theorem for $\DX$.


Let $X=((x_1^1,x_1^2),\ldots,(x_n^1,x_n^2))$ be a sequence of $n$ pairs of non-negative elements, and we define $X^c\in\R^{2n}$ for $c\in\left\{1,2\right\}^n$ by $\left(X^c\right)_i=x_i^{c_i}$.
We define a distribution $\DX$ over $\Reals^{2n}$ to be the uniform distribution over $\left\{X^c\right\}_{c\in\left\{1,2\right\}^n}$.
In addition, we define $\DXaward$ to be the projection of $\DX$ on the last $n$ coordinates (the awards).

\begin{lemma} \label{lem:key}
	Let $\ALG$ be a single-sample algorithm.
	If for every $X \in \R^{2n}_{\geq0}$ it holds that
	$$\TOPellVal(\ALG(\DX)) \geq \rho \cdot \TOPellVal(\DXaward), $$
	then for every product distribution $\dist=\dist_1 \times \ldots \times \dist_n$,
	$$\TOPellVal(\ALG(\dist^2)) \geq \rho \cdot \TOPellVal(\dist).$$
\end{lemma}

Hence, in order to prove Theorem~\ref{thm:alg-performance} it suffices to show that for every $X$, $\ALG^\tau$ performs well when its input is drawn according to $\DX$.
\begin{proof}
	First we note that if for every $X$
	it holds that 
	$$\TOPellVal(\ALG(\DX)) \geq \rho \cdot \TOPellVal(\DXaward),$$
	then
	$$
	\Exp_{X \sim \dist^2}\left[\TOPellVal(\ALG(\DX))\right]\geq  \rho \cdot \Exp_{X \sim \dist^2} \left[ \TOPellVal(\DXaward) \right].$$
	
	Since the distribution of the input of $\ALG$ (i.e., $Y\sim\DX$ where $X\sim\dist^2$) is equivalent to $\dist^2$, the left hand side of the inequality equals to $\TOPellVal(\ALG(\dist^2))$.
	Similarly, $\DXaward$ where $X\sim\dist^2$ is equivalent to $\dist$, so the right hand side equals $\rho \cdot \TOPellVal(\dist)$.
	Hence, we get that 
	$\TOPellVal(\ALG(\dist^2)) \geq \rho \cdot \TOPellVal(\dist)$.
\end{proof}

It now remains to prove that for every $X \in \R^{2n}_{\geq0}$, $\ALG^{\tau}$ achieves a good approximation with respect to $\DX$.

\begin{lemma}
	\label{lem:sufficient}
	For every $X \in \R^{2n}_{\geq0}$, it holds that
	$$
	\TOPellVal(\ALG^{\tau}(\DX)) \geq \left(1-4\ell\cdot e^{-\frac{1}{8k}\cdot\left(\min\left(k-\tau,\tau-\ell\right)\right)^{2}}\right)\cdot \TOPellVal(\DXaward).
	$$
	
\end{lemma}

\begin{proof}
	Given $X$, the random tie-breaking rule is equivalent to an infinitesimal perturbation of the entries of $\left(\samples,\vals\right)$ of the form $\delta\cdot\epsilon$ for $\epsilon\in\left(-1,1\right)^{2n}$ being an i.i.d. $U\left(-1,1\right)$ noise and $\delta>0$ being infinitesimal compared to the non-zero elements of $X$. Note that this perturbation is equivalent to perturbing the entries of $X$ by $\delta\cdot\epsilon$.
	Since such perturbation does not change (in an essential way) the value of subsets of $X$, we can assume w.l.o.g. that all the entries of $X$ are different from each other (and hence the elements of $X$ are strictly ordered, so the tie-breaking rule plays no role and can be ignored).
\newcommand{\Item}{Y}%

Given a vector $X \in \R^{2n}_{\geq0}$, we define $\Item_i$ to be the $\Xth{i}$ highest entry in $X$, and for simplicity we define $Y_{2n+1}=0$.
\newcommand{\Value}{Aw}%
We use the notation  $\Value\left(X^c\right)$ for the last $n$ coordinates of $X^c$ (i.e., the awards).
Recall that $\DX$ was defined as the uniform distribution over $\left\{X^c\right\}_{c\in\left\{1,2\right\}^n}$.
Note that given $X^c$, $\ALG^\tau$ is deterministic and hence,
%
\begin{eqnarray}
\TOPellVal(\ALG^{\tau}(\dist_X)) \nonumber
	& = & \sum_{j=1}^{2n}\Pr\left[\Item_j \in \TOPellSet(\ALG^{\tau}( X^c))\right]\cdot \Item_j \nonumber \\
	& = & \sum_{j=1}^{2n}\Pr\left[\Item_j \in \TOPellSet(\ALG^{\tau}( X^c))\right]\cdot \sum_{i=j}^{2n}\left(\Item_i - \Item_{i+1}\right) \nonumber \\
	& = & \sum_{i=1}^{2n}\left(\Item_i - \Item_{i+1}\right)\cdot\sum_{j=1}^{i}\Pr\left[ \Item_j \in \TOPellSet(\ALG^{\tau}(X^c))\right]. \label{eq:val1}
\end{eqnarray}

Consider the summation on the right. It holds that	
\begin{eqnarray}
& & \hspace{-2em}\sum_{j=1}^{i}\Pr\left[\Item_j \in \TOPellSet(\ALG^{\tau}(X^c))
\right] \nonumber \\
& \geq  & \sum_{j=1}^{i}\Pr\left[\Item_j \in \TOPellSet(\ALG^{\tau}( X^c)) \cap \TOPellSet(\Value(X^c))
\right] \nonumber \\
& = & \sum_{j=1}^{i}\left(\Pr\left[\Item_j \in \TOPellSet(\Value(X^c)) \right]
-\Pr\left[\Item_j \in \TOPellSet(\Value(X^c)) \text{ and } \Item_j \notin \TOPellSet(\ALG^{\tau}( X^c)) \right]\right) \nonumber \\
& \geq & \sum_{j=1}^{i}\Pr\left[\Item_j \in \TOPellSet(\Value(X^c))\right]-
\min\left(i,\ell\right)\cdot\Pr\left[\TOPellSet(\Value(X^c))\neq \TOPellSet(\ALG^{\tau}( X^c))\right]. \label{minil}
\end{eqnarray}
%
The last inequality follows since whenever $\Item_j \in \TOPellSet(\Value(X^c))$ and $\Item_j \notin \TOPellSet(\ALG^{\tau}( X^c))$, it holds that $\TOPellSet(\Value(X^c)) \neq \TOPellSet(\ALG^{\tau}( X^c))$, and we sum over at most $\min\left(i,\ell\right)$ elements.


Consider the left term in Eq.~(\ref{minil}); It holds that
	$$
	\sum_{j=1}^{i}\Pr\left[\Item_j \in \TOPellSet(\Value(X^c)) \right]
	 \geqslant \Pr\left[\Item_1 \in \TOPellSet(\Value(X^c)) \right]
	 =  \frac{1}{2}\mbox{ }\geqslant \mbox{ } \frac{\min\left(i,\ell\right)}{2\ell},
	$$
where the equality holds since $\Item_1$ is the largest value, and so it is picked by the prophet if it is a reward, which occurs with probability $1/2$.
	Substituting the last inequality in Eq.~(\ref{minil}) gives
	\begin{eqnarray}
	& & \hspace{-2em}\sum_{j=1}^{i}\Pr\left[\Item_j \in \TOPellSet(\ALG^{\tau}(X^c) )\right] \nonumber \\
	& \geqslant &
\sum_{j=1}^{i}\Pr\left[\Item_j \in \TOPellSet(\Value(X^c))\right] \cdot \left(
	1-2\ell\cdot\Pr\left[\TOPellSet(\Value(X^c))\neq \TOPellSet(\ALG^{\tau}( X^c))\right]
	\right). \label{eq:val2}
	\end{eqnarray}
	Substituting in Eq. (\ref{eq:val1}), we get:
	\begin{eqnarray}
	& & \hspace{-2em}\TOPellVal(\ALG^{\tau}(\dist_X))\nonumber\\
	& \stackrel{(\ref{eq:val1})}{=} & \sum_{i=1}^{2n}\left(
	\Item_i-\Item_{i+1}\right)\cdot\sum_{j=1}^{i}\Pr\left[
	\Item_j \in \TOPellSet(\ALG^{\tau}(X^c)) \right]\nonumber \\
	&  \stackrel{(\ref{eq:val2})}{\geq} &\left(
	1-2\ell\cdot\Pr\left[\TOPellSet(\Value(X^c))\neq \TOPellSet(\ALG^{\tau}( X^c))\right]
	\right)\cdot \sum_{i=1}^{2n}\left(\Item_i -\Item_{i+1} \right)\cdot
	\sum_{j=1}^{i}\Pr\left[
	\Item_j \in \TOPellSet(\Value(X^c))
	\right] \nonumber\\
	& = & \left(
	1-2\ell\cdot\Pr\left[\TOPellSet(\Value(X^c))\neq \TOPellSet(\ALG^{\tau}( X^c))\right]
	\right)\cdot
	\sum_{j=1}^{2n}\Pr\left[
	\Item_j \in \TOPellSet(\Value(X^c))\right]\cdot
	\sum_{i=j}^{2n}\left(
	\Item_i -\Item_{i+1} \right)\nonumber\\
	& = & \left(
	1-2\ell\cdot\Pr\left[\TOPellSet(\Value(X^c))\neq \TOPellSet(\ALG^{\tau}( X^c))\right]
	\right)\cdot
	\sum_{j=1}^{2n}\Pr\left[
	\Item_j \in \TOPellSet(\Value(X^c))\right]\cdot
	\Item_j\nonumber\\
	& = &\left(
	1-2\ell\cdot\Pr\left[\TOPellSet(\Value(X^c))\neq \TOPellSet(\ALG^{\tau}( X^c))\right]
	\right)\cdot\TOPellVal(\dist_X^{award}).\label{eq:2}
	\end{eqnarray}

\newcommand{\AwardsAboveThreshold}{A^c}%
Given an index $c\in\left\{1,2\right\}^n$, we define $T^c$ to be the threshold set by $\ALG^\tau\left(X^c\right)$  and $\AwardsAboveThreshold$ to be the set of awards above the threshold, $\left\{\mbox{ }i\mbox{ } \SetSt \mbox{ } \Value(X^{c})_i > T^c\right\}$.
Note that if $\ell \leq \left|\AwardsAboveThreshold\right| \leq k$, then $\TOPellSet(\Value(X^c)) = \TOPellSet(\ALG^{\tau}( X^c))$.
Therefore,
\begin{equation}
	\Pr\left[\TOPellSet(\Value(X^c))\neq \TOPellSet(\ALG^{\tau}( X^c))\right]
	\leqslant \Pr\left[ \left|\AwardsAboveThreshold\right|<\ell\right]+
	\Pr\left[ \left|\AwardsAboveThreshold\right| > k \right]
	\label{eq:3}
\end{equation}
	
Let $\mathbb B$ be the set of the $\left(\tau+k\right)$ highest elements of $X$, and let $B$ be the set
	\begin{equation}
	B=\left\{\mbox{ } i \mbox{ }\SetSt \text{ exactly one of }\left\{ x_{i}^{1},x_{i}^{2}\right\} \text{ is in }\mathbb B\right\}\text.
	\end{equation}
Notice that if for less than $\frac{\left|B \right|-\left(k-\tau\right)}{2}$ of the pairs in $B$ the greater element of the pair was chosen to be a sample,
	then there are less than $\tau$ samples in $\mathbb B$,
	thus $\left|\AwardsAboveThreshold\right| > k$.
Since for each $i \in B$, the greater element of the pair $(x_{i}^{1},x_{i}^{2})$ is independently chosen to be a sample with probability $1/2$, by applying Chernoff bound, we get
	\begin{eqnarray}
	\Pr\left[\left|\AwardsAboveThreshold\right| > k \right] \leqslant  e^{-\frac{\left(k-\tau\right)^{2}}{4\left|B\right|}}
	\leqslant  e^{-\frac{1}{4}\cdot\frac{\left(k-\tau\right)^{2}}{k+\tau}}.\label{eq:4}
	\end{eqnarray}
Similarly,
	by setting $\mathbb B$ to be the set of the $\left(\tau+\ell \right)$ highest elements of $X$ and $B$ to be the set of indices s.t. exactly one of $\left\{ x_{i}^{1},x_{i}^{2}\right\}$ is in $\mathbb B$, and following the same analysis, we get that
	\begin{equation}
	\Pr \left[\left|\AwardsAboveThreshold \right| < \ell \right]
	 \leqslant e^{-\frac{\left(\tau-\ell\right)^{2}}{4\left|B\right|}}
	 \leqslant e^{-\frac{1}{4}\cdot\frac{\left(\tau-\ell\right)^{2}}{\tau+\ell}} \label{eq:5}.
	\end{equation}	
The desired result is now obtained by combining Equations (\ref{eq:2}, \ref{eq:3}, \ref{eq:4}, \ref{eq:5}),
	\begin{eqnarray}
	\TOPellVal(\ALG^{\tau}(\dist_X))
	& \geqslant & \left(1-
	2\ell\cdot\left(e^{-\frac{1}{4}\cdot\frac{\left(k-\tau\right)^{2}}{k+\tau}}
	+e^{-\frac{1}{4}\cdot\frac{\left(\tau-\ell\right)^{2}}{\tau+\ell}}\right)
	\right)\cdot\TOPellVal(\dist_X^{award}) \nonumber\\
	& \geqslant & \left(1-
	4\ell\cdot e^{-\frac{1}{8k}\cdot\left(\min\left(k-\tau,\tau-\ell\right)\right)^{2}}
	\right)\cdot\TOPellVal(\dist_X^{award}).\qedhere
	\end{eqnarray}
\end{proof}

We next show that in the case where $k-\ell = \Theta(\ell)$, the bound given in Theorem \ref{thm:alg-performance} is tight with respect to any single-sample algorithm.

\begin{theorem} \label{thm:single_sample_lower}
	No single-sample algorithm can achieve a competitive ratio better than $1 -\frac{2^{-(2k+1)}}{k+1}$.
\end{theorem}

\begin{proof}
	Let \ALG{} be a single-sample \lOutOfkTEXT{} algorithm and let  $\left\{L_i\right\}_{i=1}^{k+1}$ and $\left\{H_i\right\}_{i=1}^{k+1}$ be
	two sequences s.t.	\begin{equation}
				0<L_1<L_2<\cdots<L_{k+1}<H_1<H_2<\cdots<H_{k+1}\text.\end{equation}

Consider the event where the sample vector is $\samples=\left(L_1,\ldots,L_{k+1},0,\ldots,0\right)$
		and the award vector is $\vals=\left(H_1,\ldots,H_{k+1},0,\ldots,0\right)$.
	Since \ALG{} accepts at most $k$ elements, there must be some $j$, $1\leq j\leq k+1$, s.t. \ALG{} rejects  the $\Xth{j}$ element with probability at least $\frac{1}{k+1}$. Let $\bar{j}$ denote this index.
	
Suppose both samples and awards are distributed according to the following product distribution $\dist=\dist_1 \times \cdots \times \dist_n$:
\begin{itemize}
	\item For $i \leq \bar{j}$, $\dist_i$ gives either $L_i$ or $H_i$, each with probability $1/2$.
	\item For $i=\bar{j}+1,\ldots,k+1$, $\dist_i$ gives $L_i$ with probability 1.
	\item For $i=k+2,\ldots,n$, $\dist_i$ gives $0$ with probability 1.
\end{itemize}

Now, consider the event that the sample vector is $\samples$ as before, the award vector is
	\begin{equation}%
	\vals'=\left(H_1,\ldots,H_{\bar{j}},L_{\bar{j}+1},\ldots,L_{k+1},0,\ldots,0\right)\text,
	\end{equation}
	(note that both $\samples$ and $\vals'$ are in the support of $\dist$) and \ALG{} rejects the $\Xth{\bar{j}}$ element.
Since $\val_i=\val'_i$ for every $i\leq \bar{j}$, this event happens with probability at least $2^{-2\bar{j}}\cdot\frac{1}{k+1}$.
In this event, the difference between the optimal performance and the performance of \ALG{} is at least
$H_{\bar{j}}-H_{\bar{j}-1}$ (use the convention that $H_0 = L_{k+1}$), and hence
\begin{equation}
		\TOPellVal\left(\dist\right) - \TOPellVal\left(\ALG\left(\dist\right)\right)  \geq	\frac{2^{-2\bar{j}}}{k+1}\cdot\left(H_{\bar{j}}-H_{\bar{j}-1}\right).\nonumber
\end{equation}
It also holds that $\TOPellVal\left(\dist\right)\leq \frac{1}{2} H_{\bar{j}} + \ell H_{\bar{j}-1}$.
We get that the competitive ratio of \ALG{} is at most
\begin{eqnarray*}
\frac{\TOPellVal\left(\ALG\left(\dist\right)\right)}{\TOPellVal\left(\dist\right)} & = & 1 - \frac{\TOPellVal\left(\dist\right) - \TOPellVal\left(\ALG\left(\dist\right)\right)}{\TOPellVal\left(\dist\right)} \\
& \leq & 1 - \frac{\frac{2^{-2\bar{j}}}{k+1}\cdot\left(H_{\bar{j}}-H_{\bar{j}-1}\right)}{\frac{1}{2} H_{\bar{j}} + \ell H_{\bar{j}-1}} \\
& \leq & 1 - \frac{2^{-(2k+1)}}{k+1} \cdot \frac{H_{\bar{j}}-H_{\bar{j}-1}}{H_{\bar{j}} + 2\ell H_{\bar{j}-1}}.
\end{eqnarray*}
For a sufficiently large ratio of $H_{\bar{j}}/H_{\bar{j}-1}$ the bound is arbitrarily close to $1 -\frac{2^{-(2k+1)}}{k+1}$.
\end{proof}

\subsection{Algorithm Based on the Distribution of Max Value}
\label{sec:prophet_max}
In this section, we show a simple single-threshold algorithm for \OneOutOfkProphetTEXT{} scenarios which is based only on the distribution of the maximal element, and achieves
	$1-\frac{3}{2}\cdot e^{-\frac{k}{6}}$ 
fraction of the expected value of the maximal element.
Let $D_{max}$ be the distribution of the maximal element, and we assume that $D_i$ has no mass points for all $i \in [n]$.

\newcommand{\ALGmax}{\ALG_{max}}
\begin{algorithm*}[$\ALGmax$]	\mbox{ }
	\begin{itemize}
		\item Set a threshold $T$ such that    $
		\Pr_{x\sim \dist_{max}}[x < T] = \left(\frac{2}{3}\right)^{k-1}.
		$
		\item 	Accept the first $k$ elements for which $\val_{i} > T$.
		
	\end{itemize}
\end{algorithm*}

\begin{theorem}\label{thm:alg_max}
For every product distribution $\dist$ and $\vals \sim \dist$:
	The expected value of the maximal element returned by $\ALGmax$ is at least
		 $1-\frac{3}{2}\cdot e^{-\frac{k}{6}}$ 
	of the expected maximal element. I.e., $\ALGmax$ achieves a competitive ratio of
		 $1-\frac{3}{2}\cdot e^{-\frac{k}{6}}$.
\end{theorem}

\begin{proof}
We use the notation $OPT$ for the expected maximal element.
First, we show that for every $k \leq n$ it holds that
\begin{eqnarray}
	& &\hspace{-2em}\E\left[ \sum_{j=1}^{n} \mathds{1}_{\val_j > T } \right]  =  n \cdot\left(1 - \frac{1}{n}\sum_{j=1}^{n} \Pr\left[ \val_j < T \right]\right) \nonumber \\
	& \leq & n \cdot\left(1 - \left( \prod_{j=1}^{n} \Pr\left[ \val_j < T \right]  \right) ^{1/n}\right)
	= n \left(1- \left(\frac{2}{3}\right)^\frac{k-1}{n}\right)
	\leq \frac{k}{2}, \label{eq:expectation_chernoff}
\end{eqnarray}
where the inequality follows from the inequality of arithmetic and geometric means.
We can now apply Chernoff bound to get:
\begin{eqnarray}
\Pr\left[\left( \sum_{j=1}^{n} \mathds{1}_{\val_j > T } \right) \geq k \right] \stackrel{(\ref{eq:expectation_chernoff})}{\leq} {e}^{-\frac{k}{6}}. \label{eq:probability_chernoff}
\end{eqnarray}
Next, we notice that for every $x$ it holds that
$
\Pr\left[\val_j > T  ~\SetSt~ \val_j < x\right]  \leq \Pr [\val_j > T ], \nonumber
$ 
and since $\val_j$ are independent we get that for any prefix of size $\left(i-1\right)$,\footnote{%
	If $\left\{A_i\right\}_{i=1}^n$ are independent events, and
	$\left\{B_i\right\}_{i=1}^n$ are independent events,
	and for all $i$~~$\Pr\left[A_i\right]\leq\Pr\left[B_i\right]$,
	then $
	\Pr\left[\left( \sum_{i=1}^n \mathds{1}_{A_i}\right)< k\right]
	\geq \Pr\left[\left( \sum_{i=1}^n \mathds{1}_{B_i} \right)< k \right]		
	$.%
}%
\begin{eqnarray}
\Pr\left[\left( \sum_{j=1}^{i-1} \mathds{1}_{\val_j > T}\right)< k~\SetSt~\text{for }j< i,~\val_j<x\right]
&\geq & \Pr\left[\left( \sum_{j=1}^{i-1} \mathds{1}_{\val_j > T } \right)< k \right]\nonumber\\
&\geq & \Pr\left[\left( \sum_{j=1}^{n} \mathds{1}_{\val_j > T } \right)< k \right]\nonumber\\
& \stackrel{(\ref{eq:probability_chernoff})}{\geq}  & 1-e^{-\frac{k}{6}}.\label{eq:remove_if}
\end{eqnarray}
Hence, also
\begin{equation}
\Pr\left[
\left( \sum_{j=1}^{\argmax\left(\vals\right)-1} \mathds{1}_{\val_j > T } \right)< k
~\SetSt~\max\vals=x\right]
\stackrel{(\ref{eq:remove_if})}{\geq} 1-e^{-\frac{k}{6}}.\label{eq:remove_if_max}
\end{equation}

We are now ready to establish the bound on the performance of $\ALGmax$.
\begin{eqnarray}
\E[\ALGmax]
& \geq &  \E_{x \sim \dist_{max}}\left[\mathds{1}_{x>T}\cdot x\cdot
\Pr\left[\ALG \text{ picked the maximal element}
~\SetSt~\max\vals=x\right]\right]  \nonumber \\
& = &  \E_{x \sim \dist_{max}}\left[\mathds{1}_{x>T}\cdot x\cdot
\Pr\left[
\left( \sum_{j=1}^{\argmax\left(\vals\right)-1} \mathds{1}_{\val_j > T } \right)< k
~\SetSt~\max\vals=x\right]\right]  \nonumber \\
& \stackrel{(\ref{eq:remove_if_max})}{\geq}   &  \left(1-e^{-\frac{k}{6}}\right) \cdot \E_{x \sim \dist_{max}}\left[\mathds{1}_{x>T}\cdot x\right]\nonumber \\
& \geq  &  \left(1-e^{-\frac{k}{6}}\right) \cdot \Pr\left[x>T\right]\cdot\E_{x \sim \dist_{max}}\left[x\right]\nonumber \\
& =  &  \left(1-e^{-\frac{k}{6}}\right) \cdot \left(1 - \left(\frac{2}{3}\right)^{k-1} \right) \cdot OPT
\quad  \geq   OPT \cdot \left(1-\frac32\cdot e^{-\frac{k}{6}}\right). \nonumber\qedhere  \end{eqnarray}
\end{proof}

We observe that a single-threshold algorithm, as we analyzed here, can be translated into an algorithm in the setting of~\cite{AssafSC2000} as follows. Our algorithm sets a threshold $T$ and pick the first $k$ elements that exceed $T$. To apply our algorithm to their setting, let \ALG$^{i}$ (for $i=1, \ldots, k$) be the algorithm that accepts the $\Xth{i}$ element that exceeds $T$. One can easily verify that a competitive ratio for our setting carries over to their setting. 
Our result (Thm.~\ref{thm:alg_max}) shows that with exactly the same limited information as~\cite{AssafSC2000}, $\dist_{max}$, one can get an improved competitive ratio,
		 $1-\frac{3}{2}\cdot e^{-\sfrac{k}{6}}$, 
that decreases exponentially with $k$.\footnote{Note that our algorithm translates to $k$ semi-threshold algorithms, where the $\Xth{i}$ element that exceeds the threshold is selected, rather than the first one.}

We also note that algorithm $\ALGmax$ can be modified slightly to handle cases in which $\dist_i$ might contain mass points.\footnote{The modified algorithm is not a single-threshold algorithm, but it can still be simulated by $k$ simultaneous algorithms.}
In order to do so, we make the following modifications:
(a)~instead of setting $T$ such that $\Pr_{x\sim \dist_{max}}[x < T] = \left(\frac{2}{3}\right)^{k-1}$, we set $T = \inf\{\mbox{ } t \mbox{ }| \mbox{ } \Pr_{x\sim \dist_{max}}[x \leq t] \geq \left(\frac{2}{3}\right)^{k-2}\}$.
(b)~the algorithm accepts the first element $v_i$ such that $v_i \geq T$.
(c)~thereafter, the algorithm accepts the first $k-1$ elements $v_i$ such that $v_i>T$.
A similar analysis shows that the modified algorithm guarantees a competitive ratio of
	$1-\frac{3}{2}\cdot e^{-\frac{k-1}{6}}$.

The following theorem gives an impossibility result for any online algorithm, even ones that know the distributions $\dist_i$.

\begin{theorem} \label{thm:prophet_example}
There exists a product distribution $\dist$ such that there is no algorithm $\ALG$ that achieves a competitive ratio better than $1 - \frac{1}{(2k+2)!}$. \textit{I.e.}, for every $\ALG$ it holds that
	$$\TOPellVal(\ALG(\dist)) \leq \left( 1 - \frac{1}{(2k+2)!} \right) \cdot \TOPellVal(\dist).$$	
\end{theorem}

\begin{proof}
	Let $\dist_i$ be the distribution that assigns the value $x_i = \frac{i\cdot (2k-i+3)}{2}$ with probability $\frac{1}{x_i}$, and $0$  otherwise.
	Let $\dist_0$ be the distribution that gives a value $0$ with probability $1$.
	And let $\dist = \dist_1 \times \cdots \times \dist_{k+1} \times \dist_0 \times \ldots \times \dist_0$.
	We first claim that the best algorithm for this distribution must accept any non-zero element, as long as it is feasible (i.e., it accepted at most $k-1$ elements so far).
	To see that, note that accepting a non-zero element in iteration~$i$ increases the performance of $\ALG$ by at least
	$$\frac{i\cdot (2k-i+3)}{2}- \frac{(i-1)(2k-i+4)}{2}= k-i+2.$$
	This follows since the $\Xth{i}$  element conditioned on being non-zero equals $\frac{i\cdot (2k-i+3)}{2}$, while the performance lost due discarding one of the former elements is bounded from above by $\frac{(i-1)(2k-i+4)}{2}$.
	On the other hand, the increase in performance due to discarding the $\Xth{i}$ element is bounded by the sum of the expected values of the remaining elements, which equals
	$$\sum_{j=i+1}^{k+1}\Exp\left[\val_j\right]=k+1-i < k-i+2.$$
	Therefore, always accepting any non-zero element while feasible is optimal.
	
	Let $\ALG$ be the optimal algorithm for the distribution $\dist$ described above (i.e., one that always accepts any non-zero element), and consider the event where $\val_i \neq 0$ for every $i \in [k+1]$.
	By the characterization above, the performance of $\ALG$ is $\sum_{i=k-\ell+1}^{k}x_i$, while the performance of the prophet is
	$\sum_{i=k-\ell+2}^{k+1}x_i$.
	Therefore, the difference between the two in this event is at least $x_{k+1} - x_{k-\ell+1} \geq x_{k+1} - x_k = 1$.
	It follows that
	\begin{eqnarray}
	\TOPellVal(\dist)-\TOPellVal(\ALG(\dist)) & \geq & \Pr\left[\forall i\in [k+1]~\val_i \neq 0\right] \cdot 1  \nonumber\\
	& = &  \prod_{i \in [k+1]}\frac{2}{i\cdot (2k-i+3)}   
	  =     \frac{2^{k+1}}{(2k+2)!}.
	 \label{eq:dif_lower_bound}\end{eqnarray}
	In addition, it holds that
	\begin{eqnarray}
	\TOPellVal(\dist) \leq \TOPenVal(\dist)=\sum_i\Exp\left[\val_i\right]=k+1.
	\label{eq:sum_lower_bound}
	\end{eqnarray}
Combining Equations (\ref{eq:sum_lower_bound}) and (\ref{eq:dif_lower_bound}) gives us
$$\TOPellVal(\ALG(\dist)) \leq \left(1 - \frac{2^{k+1}}{(k+1)(2k+2)!}\right) \cdot \TOPellVal(\dist) \leq \left(1 - \frac{1}{(2k+2)!}\right) \cdot \TOPellVal(\dist).\qedhere$$
\end{proof} 

\section{$\ell$-out-of-$\protect\DoNotCapitalizeK$ Secretaries}
\label{sec:newsecretary}
Let $\vec{\beta} = (\beta_{-1}=0,\beta_0,\ldots, \beta_{\ell} =n)$ be a vector
such that $\beta_j \leq \beta_{j+1}$ for all $j$.
We think of $\vec{\beta}$ as a partition of $[n]$ into $\ell+1$ intervals, where interval $j=0,\ldots,\ell$ is
$I_j = [\beta_{j-1}+1, \beta_j]$.

Given a vector ${\vec{\beta}}$, we define a function $b_{\vec{\beta}}:[n] \rightarrow \{0,\ldots,\ell\}$ that receives an index $i \in [n]$ and returns the unique value $j \in [0,\ldots , \ell]$ such that $i \in I_j$.
That is $b_{\vec{\beta}}(i)$ is the index of the interval that contains the $i^{th}$ element.
When clear from the context, we omit the subscript $\vec{\beta}$ and write simply $b$.

For example, suppose $n=8, \ell = 3$, and $\vec{\beta}=(0,1,4,4,8)$.
In this case, $I_0=[1,1], I_1=[2,4], I_2=\emptyset, I_3=[5,8]$.
Consequently, $b(1)=0$, $b(2)= b(3) = b(4) = 1$, and  $b(5)= b(6) = b(7) = b(8) = 3$.

Consider the following algorithm for the \lOutOfkSecretaryTEXT{} problem:

\begin{algorithm*}[$\ALG^{\vec{\beta}}(\val_1,\ldots,\val_n)$]
	\mbox{ }
	\begin{itemize}
	\item	 For $i=1,\ldots,n$, accept $\val_i$ if it belongs to the set of the $b_{\vec{\beta}}(i)$ highest elements among $\val_1,\ldots,\val_i$ and less than $k$ elements were accepted so far.
	\end{itemize}
\end{algorithm*}

\begin{example}
\label{ex:beta}
Suppose $n=8, $\vals = (2,9,3,5,4,7,6,10)$, \ell = 3, k =4$, and $\vec{\beta}=(0,1,4,4,8)$.
For $i=1~ (\val_1=2)$, $\ALG^{\vec{\beta}}(\vals)$ does not accept $\val_1$ since $b(1)=0$.
For $i=2~ (\val_2=9)$, $\ALG^{\vec{\beta}}(\vals)$ accepts $\val_2$ since $b(2)=1$, and $\val_2$ is the highest value so far.
For $i=3,4~ (\val_3=3,~\val_4=5)$, $\ALG^{\vec{\beta}}(\vals)$ does not accept $\val_i$ since $b(i)=1$, and $\val_i$ is not the highest value so far.
For $i=5,6,7~ (\val_5=4,~\val_6=7,~val_7=6)$, $\ALG^{\vec{\beta}}(\vals)$ accepts $\val_i$ since $b(i)=3$, and for each of these $i$'s, $\val_i$ is among the 3 highest values so far.
For $i=8~ (\val_8=10)$, $\ALG^{\vec{\beta}}(\vals)$ does not accept $\val_{8}$ since $k=4$ values were already accepted by $\ALG^{\vec{\beta}}$.
\end{example}


For the remainder of this section, let $s=\frac{k-8\ell}{2+2\ln \ell} \geq 0$.
We study the performance of $\ALG^{\vec{\beta}}$ with the following parameters:
\begin{equation}
\beta_0 = \lfloor \frac{ne^{-s}}{2e\ell} \rfloor \quad\quad \mbox{and} \quad\quad \beta_j= \lfloor \frac{jne^{-s/j}}{2e\ell} \rfloor \quad \mbox{ for } 0<j<\ell.
\label{eq:betas}
\end{equation}

\begin{theorem}\label{thm:secretary_b_new}
For large enough $n$, $\ALG^{\vec{\beta}}$ with the parameters above guarantees a competitive ratio of $1-\ell e^{-\frac{k-8\ell}{2+2\ln \ell}} - e^{-k/6}$.
\end{theorem}

Note that the dominant term in this bound is $e^{-k/6}$ for $\ell < 8 $, and $e^{-\frac{k-8\ell}{2+2\ln \ell}}$ for $\ell \geq 8$.
In order to analyze the performance of $\ALG^{\vec{\beta}}$ we first consider the following simpler algorithm which might accept more than $k$ elements:
\begin{algorithm*}[$\ALG'^{\vec{\beta}}(\val_1,\ldots,\val_n)$]
	\mbox{ }
	\begin{itemize}
			\item For $i=1,\ldots,n$, accept $\val_i$ if it belongs to the set of the $b_{\vec{\beta}}(i)$ highest elements among $\val_1, \ldots, \val_i$.
	\end{itemize}
\end{algorithm*}

In what follows we analyze the performance of $\ALG'^{\vec{\beta}}$ and then show that $\ALG^{\vec{\beta}}$ and $\ALG'^{\vec{\beta}}$ return the same set of values with high probability.

%
%
\begin{lemma} \label{l:exp_last_new}
	For large enough $n$, for any vector $\vals=(\val_1,\ldots,\val_n)$, $$\TOPellVal\left(\ALG'^{\vec{\beta}}\left(\vals\right)\right)\geq \left(1-\ell e^{-s}\right) \cdot \TOPellVal(\vals).$$
\end{lemma}
\begin{proof}
	The probability that the $\Xth{i}$ highest value for $i \leq \ell$ is accepted by $\ALG'^{\vec{\beta}}$ is at least the probability that the $\Xth{\ell}$ highest value is accepted by $\ALG'^{\vec{\beta}}$, since the arrival order is uniform. Therefore, a lower bound $\alpha$ for the probability that the $\Xth{\ell}$ highest value is accepted immediately implies a lower bound of $\alpha$ on the competitive ratio of $\ALG'^{\vec{\beta}}$.
	
Let $\val$ be the $\Xth{\ell}$ highest value in $\vals$. In the following inequalities, $HG$ and $Bin$ stand for the hypergeometric and Binomial distributions, respectively. Justifications for the following derivations are given below.
\begin{eqnarray}
& &\Pr\left[\val \mbox{ is not selected by } \ALG'^{\vec{\beta}}\right] \nonumber\\
&=& \sum_{i=1}^n \Pr\left[\val = \val_i\right] \cdot \Pr\left[\val \notin \TOPSet^{b(i)}(\val_1,\ldots,\val_i)~|~\val=\val_i \right] \label{eq_total}\\
& =& \frac{1}{n} \sum_{i=1}^n \Pr[HG(n-1,\ell-1,i-1) \geq b(i)] \label{eq_hg} \\
& \leq & \frac{1}{n} \sum_{i=1}^n \Pr[HG(n-1,\ell-1,\beta_{b(i)}) \geq b(i)] \label{eq_ibi} \\
& = & \frac{1}{n} \sum_{j=0}^{\ell-1} (\beta_j-\beta_{j-1}) \cdot \Pr[HG(n-1,\ell-1,\beta_j) \geq j] \nonumber \\
& \leq & \frac{1}{n} \sum_{j=0}^{\ell-1} \beta_j \cdot \Pr[HG(n-1,\ell-1,\beta_j) \geq j] \nonumber \\
& \leq & \frac{1}{n} \sum_{j=0}^{\ell-1} \beta_j \cdot \Pr[Bin(\beta_j,\frac{\ell-1}{n-1-\beta_j}) \geq j] \label{eq_hg2} \\
& \leq & \frac{\beta_0}{n} + \frac{1}{n} \sum_{j=1}^{\ell-1} \beta_j \cdot \left(\frac{e\beta_j(\ell-1)}{j(n-1-\beta_j)}\right)^{j-1} \label{eq_chernoff} \\
& \leq & \frac{\beta_0}{n} +  \sum_{j=1}^{\ell-1}  \left(\frac{2e\beta_j \ell}{jn}\right)^{j}  \leq e^{-s} +  (\ell-1) \cdot e^{-s}= \ell \cdot e^{-s} \label{eq_alg'b}.
\end{eqnarray}
Eq. (\ref{eq_total}) follows by the law of total probability.
Eq. (\ref{eq_hg}) follows by the definition of HG distributions. Eq. (\ref{eq_ibi}) holds since $i \leq \beta_{b(i)}$.
Eq. (\ref{eq_hg2}) holds by the fact that $\Pr\left[HG(a,b,c) \geq d \right] \leq \Pr\left[Bin(c,\frac{b}{a-c}) \geq d\right]$ for all $a,b,c,d$.
Eq. (\ref{eq_chernoff}) follows by the Chernoff bound.
Finally, Eq. (\ref{eq_alg'b}) is derived by substituting $\beta_j$ and the observation that $n-1-\beta_j \geq \frac{n}{2}$ for $j < \ell$.
\end{proof}

\begin{lemma} \label{l:dif_algs_new}
	For large enough $n$, the probability that $\ALG^{\vec{\beta}}$ and $\ALG'^{\vec{\beta}}$ return different sets is  at most $e^{-k/6}$.
\end{lemma}
\begin{proof}
Let $A_i$ be the event that $\val_i$ belongs to the  $b(i)$ highest values among $\val_1,\ldots,\val_i$.
It holds that
	\begin{eqnarray}
	\sum_{i=1}^{n} \Pr\left[A_i\right] & = & \sum_{i=1}^{n} \frac{b(i)}{i} \label{eq:def_ai} \\
	& = & \sum_{j=2}^{\ell} \sum _{i=\beta_{j-1}+1}^{\beta_{j}}\frac{j}{i} \nonumber \\
	& \leq & \sum_{j=2}^{\ell} {j}\ln \frac{\beta_j}{\beta_{j-1}} \label{eq:harmonic} \\
	 &\leq&
	 \sum_{j=2}^{\ell-1}j\ln\left(\frac{j}{j-1} \cdot \frac{e^{\frac{s}{j-1}} }{e^{\frac{s}{j}}}\right) +
	  \ell \ln \frac{2e\ell}{(\ell-1)e^{-s/(\ell-1)}}
	 \label{eq:leq2} \\
	 	 &=&
	 \sum_{j=2}^{\ell}j\ln\left(\frac{j}{j-1} \cdot \frac{e^{\frac{s}{j-1}} }{e^{\frac{s}{j}}}\right) +
	 \ell \ln \frac{2e}{e^{-s/\ell}}
	 \nonumber \\
	 &\leq& 2\ell + \sum_{j=2}^{\ell} \frac{s}{j-1} + \ell\ln(2e^{1+s/\ell}) \label{eq:leq3} \\
	 &\leq&
	 2\ell + s\ln\ell +\ell\ln(2e^{1+s/\ell} ) \leq \frac{k}{2}. \label{eq:harmonic2}	
	\end{eqnarray}
Eq. (\ref{eq:def_ai}) follows by the fact that event $A_i$ occurs independently with probability $\frac{b(i)}{i}$. Eq. (\ref{eq:harmonic}) and (\ref{eq:harmonic2}) follow by the sum of harmonic series.
Eq. (\ref{eq:leq2}) is obtained by substituting the values of $\beta$'s from Equation~(\ref{eq:betas}) and omitting the $\lfloor \rfloor$.
Eq. (\ref{eq:leq3}) holds since $j\ln\frac{j}{j-1} \leq 2$ for every $j\geq 2$.
Finally, Inequality (\ref{eq:harmonic2}) is obtained by replacing the value of $s$. 

$\ALG^{\vec{\beta}}$ and $\ALG'^{\vec{\beta}}$ return different sets if and only if more than $k$ events out of $A_{1},\ldots,A_n$ occur.
Hence, by Chernoff inequality the probability that $\ALG^{\vec{\beta}}$ and $\ALG'^{\vec{\beta}}$ return different sets is bounded by $\Pr[\sum_{i=1}^n A_i > k]  \leq e^{-k/6}$.
\end{proof}
\begin{proof}[Proof of Theorem \ref{thm:secretary_b_new}]
	Combining Lemmas \ref{l:exp_last_new} and \ref{l:dif_algs_new} gives:
	\begin{eqnarray}
	\TOPellVal\left(\ALG^{\vec{\beta}}\left(\vals\right)\right)
	 & \stackrel{(\ref{l:dif_algs_new})}{\geq}  & \left(1-  e^{-k/6}\right)\TOPellVal\left(\ALG'^{\vec{\beta}}\left(\vals\right)\right) \nonumber\\
	 & \stackrel{(\ref{l:exp_last_new})}{\geq}& \left(1-\ell e^{-s}\right)\cdot \left(1-  e^{-k/6}\right) \cdot \TOPellVal(\vals) \nonumber\\
	 & \geq & \left(1-\ell e^{-\frac{k-8\ell}{2+2\ln \ell}} - e^{-k/6}\right)\cdot \TOPellVal\left(\vals\right)\nonumber
\end{eqnarray}\end{proof} \qedhere
Next, we show a lower bound for the \lOutOfkSecretaryTEXT{} problem.
\begin{lemma}\label{lem:sec_pro_max}
	No algorithm which accepts at most $k$ elements, accepts the maximal element with probability greater than $\left(1+\frac{1}{n}\right)\left(1-\frac{1}{e^k}\right)$.
\end{lemma}
\begin{proof}
	Let $p_i$ be the probability that $\ALG$ accepts the $\Xth{i}$ element given that it is the highest so far. The probability that $\ALG$ accepts the maximal element is 	\begin{eqnarray}
	& & \hspace{-2em}\sum_{i=1}^{n} \Pr\left[i \text{ is the maximal element}\right] \cdot \Pr\left[i \text{ is accepted}~\SetSt~i \text{ is the maximal element}\right]\nonumber\\
	& = & \frac{1}{n}\cdot  \sum_{i=1}^{n} \Pr\left[i \text{ is accepted}~\SetSt~i \text{ is the maximal element}\right]\nonumber
	= \frac{1}{n}\sum_{i=1}^{n}p_i\text,\nonumber
	\end{eqnarray}
	where the last step is due to the online nature of the problem.
	On the other hand, since \ALG{} chooses at most $k$ elements it holds that $\sum_{i=1}^{n} \frac{p_i }{i} \leq k$.
	Since $\sum_{i=1}^{n} \frac{p_i }{i} \geq \ln\left( \frac{n+1}{n+1-\sum_{i=1}^{n}p_i}\right)$, we get that the probability that $\ALG$ accepts the maximal element is at most $\left(1+\frac{1}{n}\right)\left(1-\frac{1}{e^k}\right)$.
\end{proof}

\begin{theorem}\label{thm:sec_lower}
No algorithm achieves a better competitive ratio than $\left(1+\frac{1}{n}\right)\left(1-\frac{1}{e^k}\right) < 1-\frac{1}{e^k}+\frac{2}{3n}$.
\end{theorem}
\begin{proof}
It follows immediately from Lemma \ref{lem:sec_pro_max} and setting values s.t. the ratio between the maximal element and the second maximal element is high enough.
\end{proof}


\section{Implications to Mechanism Design}
\label{sec:mechanism-design}

In this section\CR{,}{} we show how to use our \lOutOfkProphetTEXT{} threshold algorithms to derive truthful mechanisms for selling $\ell$ identical items that obtain the same welfare guarantees as the competitive ratios of the algorithms. The same is shown for revenue guarantees in the case of bidders that are identically distributed according to a regular distribution.


\vspace{0.1in}
\noindent {\bf Welfare maximization}
\vspace{0.1in}

%

Suppose that the agents are distributed according to a product distribution $\dist=\dist_1 \times \ldots \times \dist_n$.
	Let \ALG{} be a single-threshold algorithm for the \lOutOfkProphetTEXT{} problem with distribution $\dist$ and let $T$ be its threshold.
	In particular, for $\ALG=\ALG^{\tau}$, $T$ is the $\tau$-highest element in a sample $\samples\in\dist$,
	and for $\ALG=\ALGmax$, $T$ satisfies $\Pr_{x\sim \dist_{max}}[x<T] = \left(\frac{2}{3}\right)^{k-1}$.
	Consider the following two-phase mechanism $M_T$:\footnote{
		For the simplicity for description, we describe the mechanism under the assumption that for each of the agents, her value equals exactly $T$ with probability zero.
		The mechanism can be extended to handle non-atomless distributions by using the random tie-breaking rule of $\ALG^\tau$ and deciding for agents with value exactly $T$ whether to offer them a ticket or not.
		In particular, this extension preserves the truthfulness and welfare guarantee.}
\begin{itemize}
	\item Phase 1 (online): $k$ tickets for the second phase are offered sequentially to agents whose values exceed $T$.
	\item Phase 2 (offline, for agents that hold tickets from phase 1): $\ell$ items are sold using the VCG mechanism with reserve $T$; that is, an agent who gets an item pays the maximum between $T$ and the $\left(\ell+1\right)$-highest value among agents in the second phase.
\end{itemize}

Remark: an equivalent description of the mechanism would be one in which in phase 1, $k$ tickets are sold with uniform price $T$, and in phase 2, all agents that participate in phase 2 get reimbursement of $T$.

\begin{theorem} \label{thm:md_sw}
For every single-threshold \lOutOfkProphetTEXT{} algorithm $\ALG$, the mechanism $M_{T\left(\ALG\right)}$ is truthful, and obtains at least a fraction $\rho$ of the optimal welfare, where $\rho$ is
\ALG{}'s competitive ratio.
\end{theorem}
\begin{proof}
(truthfulness) We show that it is a dominant strategy for an agent to report truthfully in both phases.
For values $v_i \leq T$, an agent cannot benefit from receiving a ticket to the second phase, due to the fact that if she gets an item, she pays at least the reserve ($T$), and therefore she has a non-positive utility.
For values $v_i > T$, it is dominant for an agent to bid truthfully in the second phase (since the VCG mechanism is truthful), and hence dominant to accept a ticket at the first phase, since accepting a ticket guarantees her a non-negative utility.

(competitive ratio)
Due to the truthfulness of the mechanism, only agents with values $v_i > T$ will proceed to the second phase. In the second phase, the top $\ell$ values will be chosen by the VCG mechanism. Since $\ALG$ guarantees that the sum of the top $\ell$ elements chosen has a competitive ratio of $\rho$, the same guarantee holds for the welfare of the mechanism.
\end{proof}

\vspace{0.1in}
\noindent {\bf Revenue maximization}
\vspace{0.1in}
Suppose that the agents are identically and independently distributed according to a regular distribution $F$.
\newcommand{\Vvirt}{\widehat{v}}%
\newcommand{\FvirtPos}{\widehat{F}^+}%
\newcommand{\TvirtPos}{\widehat{T}^+}%
\newcommand{\pMon}{\hat{p}}%
Given a valuation $v$ distributed according to $F$, let $\Vvirt=v-\frac{1-F\left(v\right)}{f\left(v\right)}$
be the virtual valuation corresponding to $v$ (as defined by~\cite{Myerson81}), and let
$\FvirtPos$ be the distribution of $\max\left(\Vvirt,0\right)$.
Let $\pMon$ be Myerson's monopoly price for the distribution $F$ (i.e., the value corresponding to a virtual valuation of zero; $0=\pMon-\frac{1-F\left(\pMon\right)}{f\left(\pMon\right)}$.)

Let \ALG{} be a single-threshold algorithm for the \lOutOfkProphetTEXT{} problem with distribution $\times_{i=1}^n\FvirtPos$, and let $\TvirtPos$ be its threshold.\footnote{Since all the elements are non-negative with probability $1$, we assume w.l.o.g. that also the threshold $\TvirtPos$ is non-negative.}

Define $T$ to be the value corresponding to the virtual value $\TvirtPos$. I.e.,  $\TvirtPos=T-\frac{1-F\left(T\right)}{f\left(T\right)}$.
In particular, for $\ALG=\ALG^{\tau}$, $T$ is the maximum between $\pMon$ and the $\tau$-highest element in a sample $\samples\in F^n$,
and for $\ALG=\ALGmax$, $T$ is the maximum between $\pMon$ and the value $t$ satisfying $\Pr_{x\sim \dist_{max}}[x<t] = \left(\frac{2}{3}\right)^{k-1}$.
Then, run the two-phase mechanism $M_T$. 


\begin{theorem} \label{thm:md_rev}
For every single-threshold \lOutOfkProphetTEXT{} algorithm $\ALG$, the mechanism $M_{T\left(\ALG\right)}$ is truthful, and obtains at least a fraction $\rho$ of the optimal revenue, where $\rho$ is 
\ALG{}'s competitive ratio.
\end{theorem}
\begin{proof}
(truthfulness) 
The proof of truthfulness is identical to the proof given in Theorem \ref{thm:md_sw}, since truthfulness is unrelated to the objective function.

(competitive ratio)
Due to the truthfulness of the mechanism, only agents with values $v_i > T$ will proceed to the second phase. In the second phase, the top $\ell$ values will be chosen by the VCG mechanism.
Since the agents are identically distributed according to a regular distribution, these agents also have the maximal virtual valuations $\Vvirt$.
Hence, due to the competitive ratio guarantee of \ALG, $M_T$ achieves at least a fraction $\rho$  of the  expected sum of virtual valuations of the top $\ell$ agents (among all agents),  which equals to  the optimal revenue.
\end{proof}


\bibliographystyle{ACM-Reference-Format}
\bibliography{bib-file}


\end{document}